\documentclass[aip,jcp,reprint]{revtex4-1}

\usepackage{tikz}
\usepackage{amsmath, amsthm, amssymb,bbold}
\usepackage{graphicx,color,psfrag}
\usepackage[normalem]{ulem}
\graphicspath{{./Figures/}}

\usepackage[arrowdel]{physics}

\newtheorem{theorem}{Theorem}

\newtheorem{corollary}{Corollary}

\newtheorem{definition}{Definition}

\begin{document}

\title{The third law of thermodynamics in open quantum systems}
\author{Abhay Shastry}
\thanks{These two authors contributed equally.}
\affiliation{Department of Physics, University of Arizona, 1118 East Fourth Street, Tucson, AZ 85721, USA}
\affiliation{Department of Chemistry, University of Toronto, 80 St.\ George Street, Toronto, Ontario M5S 3H4, Canada}
\author{Yiheng Xu}
\thanks{These two authors contributed equally.}
\affiliation{Department of Physics, University of Arizona, 1118 East Fourth Street, Tucson, AZ 85721, USA}
\affiliation{Department of Physics, University of California, San Diego, California 92093, USA}
\author{Charles\ A.\ Stafford}
\affiliation{Department of Physics, University of Arizona, 1118 East Fourth Street, Tucson, AZ 85721, USA}

\date{\today}

\begin{abstract}
We consider open quantum systems consisting of a finite system of independent fermions with arbitrary Hamiltonian coupled to one or more equilibrium
fermion reservoirs (which need not be in equilibrium with each other).
A strong form of the third law of thermodynamics, $S(T) \rightarrow 0$ as $T\rightarrow 0$, 
is proven for {\em fully open quantum systems} in thermal equilibrium with their environment, defined as systems where 
all states are broadened due to environmental coupling.
For generic open quantum systems, it is shown that $S(T)\rightarrow g\ln 2$
as $T\rightarrow 0$, where $g$ is the number of localized states lying exactly at the chemical potential of the reservoir.
For driven open quantum systems in a nonequilibrium steady state, it is shown that the local entropy $S({\bf x}; T) \rightarrow 0$ as
$T({\bf x})\rightarrow 0$, 
except for cases of measure zero arising due to localized states,
where $T({\bf x})$ is the temperature measured by a local thermometer. 
\end{abstract}

\maketitle

\section{Introduction}
\label{sec:intro}

There are
two formulations of the third law of thermodynamics, both due to Nernst: (A) The Nernst Heat Theorem, which states that
the equilibrium entropy of a pure substance goes to zero at zero temperature, and (B) The Unattainability Principle, which
states that it is impossible to cool any system to absolute zero in a finite number of operations.

The third law in open quantum systems has been discussed in various contexts.
Kosloff and collaborators\cite{Levy2012}
consider the Unattainability Principle (B) in the context of quantum absorption refrigirators,
and show that it is not possible to cool to absolute zero in finite time. The authors
warn that the quantum Master equation has to be used carefully, and that violations of the laws of thermodynamics
could result otherwise.
The unattainability principle was challenged in Ref.\ \onlinecite{Kolar2012}, with claims that zero temperature can be reached
but that formulation (A) still holds true. Ref.\ \onlinecite{Cleuren2012}
arrives at a result which is in violation of the unattainability principle as pointed out by
a comment by Kosloff.\cite{Levy2012c}
Kosloff generally advocates for a careful use of the Master equation and in Ref.\ \onlinecite{Kosloff2013}
argues that such apparent violations\cite{Kolar2012} of the laws of thermodynamics are caused by uncontrolled
approximations.
Ref.\ \onlinecite{Masanes2017} provides a proof of the unattainability principle (B) using quantum resource theory,
and clarifies its connection to the heat theorem (A).

Statement (A) of the third law was proven for              
a quantum oscillator in contact with various types of heat baths in 
Refs.\ \onlinecite{Ford2005} and \onlinecite{OConnell2006}. 
In Ref.\ \onlinecite{OConnell2006}, O'Connell rebuts
the early claims of the violations of the laws of thermodynamics made in the field of quantum
thermodynamics. In particular, he
focuses on Ref.\ \onlinecite{Nieuwenhuizen2002}, which claims to construct a perpertual motion
machine. As relates to the third law, Ref.\ \onlinecite{Nieuwenhuizen2002} argues that ``neither  the  von  Neumann  entropy  nor  the
Boltzmann entropy vanishes when the bath temperature is zero,'' leading to a claim of
a violation of the third law for nonweak coupling.
O'Connell calculates the von Neumann entropy\cite{OConnell2006} and points out that when the interaction energy is considerable,
the von Neumann formula can only be applied to the entire system and not to the reduced system.

Perhaps
the most flagrant violation of the third law of thermodynamics was put forward by Esposito, Ochoa, and Galperin,\cite{Esposito2015}
who claim not only that Nernst's heat theorem does not hold, but that the entropy of an open quantum system is {\em undefined} in the limit
of zero temperature.
Their approach is inspired by that of Sanchez and coworkers,\cite{Ludovico2014,Ludovico2018}
who argue that the definition
of heat in open quantum systems is ambiguous when
time-dependent driving is present. In Ref.\ \onlinecite{Ludovico2014}, they argue 
that in non-steady states, the tunneling region has some
energy (``energy reactance") and it is unclear whether to
ascribe that to the ``system" or the ``bath." This, they argue, leads to an ambiguity
in the definition of the heat, which they propose is fixed by ascribing
half the energy of the tunneling region to the bath. Although
this definition of heat does agree with the laws of thermodynamics, it 
is not clear whether their prescription is applicable to models other than the
one they consider.
Ref.\ \onlinecite{Ludovico2018} proposes to experimentally
measure this energy reactance in order to determine empirically what partitioning leads to the correct heat current definition.

In Ref.\ \onlinecite{Bruch2016}, Nitzan and collaborators consider a driven resonant level model, and show that
the problem of separately defining ``system'' and ``bath'' in the strong-coupling
regime is circumvented by considering as the system everything that is influenced by the externally driven energy level, and
rebut the claims of a violation of the heat theorem put forward in Ref.\ \onlinecite{Esposito2015}.
Their book-keeping\cite{Bruch2016} is
similar to that originally put forward in the equilibrium case by Friedel.\cite{Friedel1958}
Finally, Ref.\ \onlinecite{Bruch2018} uses a scattering approach to similarly circumvent the problem of system/bath definitions for adiabatically driven 
open quantum systems, expressing changes in the entropy of the system in terms of asymptotic observables at infinity.

In this article, we investigate the applicability of Nernst's heat theorem (A) to a general class of open quantum systems, consisting of
a finite system of independent fermions with arbitrary Hamiltonian coupled to one or more equilibrium fermion reservoirs (which need not be
in equilbrium with each other).  Both the equilibrium case and the case of a nonequilibrium steady state are considered.  We consider a general
partitioning of the entire system into subsystem and reservoir(s), where the subsystem can be any finite subspace of the total Hilbert space.
We show that no ambiguity arises from the partitioning either in equilibrium or in a nonequilibrium steady state, and give proofs of
the heat theorem for both cases.

\section{Theoretical methods}
\label{sec:theory}

We consider a generic open quantum system of independent fermions coupled to one or more macroscopic fermion reservoirs.  The reservoirs are 
separately in thermodynamic equilibrium, but need not be in equilibrium with each other.

\subsection{Hamiltonian}
\label{sec:Hamiltonian}

The Hamiltonian is
\begin{equation}
H=H_{\rm sys} + H_{\rm res} + H_{\rm s-r},
\label{eq:H}
\end{equation}
where 
\begin{equation}
H_{\rm sys}=\sum_{i,j}\left(H_{\rm sys}\right)_{ij} d^\dagger_i d_j
\label{eq:H_sys}
\end{equation} 
is a generic 1-body Hamiltonian for a finite spatial domain, with $\left(H_{\rm sys}\right)^\ast_{ij}=\left(H_{\rm sys}\right)_{ji}$,
\begin{equation}
H_{\rm res} = \sum_{\alpha=1}^M \sum_{k\in\alpha} \varepsilon_k c^\dagger_k c_k
\label{eq:H_res}
\end{equation}
is the Hamiltonian describing $M$ fermion reservoirs, and
\begin{equation}
H_{\rm s-r}=\sum_i \sum_{\alpha=1}^M \sum_{k\in\alpha} \left(V_{ik} d^\dagger_i c_k + {\rm H.c.}\right)
\label{eq:H_sr}
\end{equation}
describes the system-reservoir coupling. Here $d_i$ and $c_k$ are fermion annihilation operators obeying canonical anticommutation relations.
For simplicity and mathematical rigor, we assume that the Hilbert space on which $H_{\rm sys}$ acts is {\em finite}. 

We note that the partitioning of the total system into ``system'' and ``reservoir(s)'' is to some extent arbitrary, except that the system is finite and the
reservoirs are typically taken to be infinite.  The division between system and reservoirs should be understood as a division of the Hilbert space,
not as a division of the Hamiltonian.  The open quantum system, coupled to its reservoir(s), is no longer described a hermitian Hamiltonian, but instead
is described in terms of Green's functions.

\subsection{Green's functions}
\label{sec:NEGF}

The dynamics of the open quantum system 
(\ref{eq:H}) are described by the retarded Green's function \cite{Stefanucci13,Bergfield09}
\begin{equation}
G_{ij}(t) = -i\theta(t) \langle \{d_i(t), d^\dagger_j(0)\}\rangle.
\label{eq:G_def}
\end{equation}
In this article, we consider systems in equilibrium or in a nonequilibrium steady state, so it is useful to work with its Fourier transform $G(\omega)$,
given by\cite{Stefanucci13,Bergfield09}
\begin{equation}
G(\omega) = \left[\mathbb{1}\omega - H_{\rm sys} - \Sigma(\omega)\right]^{-1},
\label{eq:G_result}
\end{equation}
where the retarded self-energy 
\begin{equation}
\Sigma(\omega)\equiv\sigma(\omega) - i\Gamma(\omega)/2 -i \eta \mathbb{1}
\label{eq:Sigma_def}
\end{equation}
accounts for the system-reservoir coupling, with $\sigma(\omega)=
\sigma^\dagger(\omega)$ and $\Gamma(\omega)=\Gamma^\dagger(\omega)$.
Here $\eta=0^+$ is a positive infinitessimal and $\mathbb{1}$ is the unit operator.
The broadening of the states of the system is determined by
\begin{equation}
\Gamma(\omega)=\sum_{\alpha=1}^M \Gamma^\alpha(\omega),
\label{eq:Gamma}
\end{equation}
where 
\begin{equation}
\Gamma_{ij}^\alpha(\omega) = 2\pi \sum_{k\in\alpha} V_{ik} V^\ast_{jk} \delta(\omega-\varepsilon_k)
=2\pi \left.\overline{V_{ik} V^\ast_{jk}}\right|_{k\in\alpha} \rho_\alpha(\omega)
\label{eq:Gamma_alpha}
\end{equation}
is the partial width function due to coupling with reservoir $\alpha$, where $\rho_\alpha(\omega)$ is the density of states of reservoir $\alpha$.
The real part of the shifts of the system energy levels due to coupling with the reservoirs is determined by
\begin{equation}
\sigma_{ij}(\omega) = \sum_{\alpha=1}^M \sum_{k\in\alpha} {\cal P} 
\frac{V_{ik} V^\ast_{jk}}{\omega-\varepsilon_k}, 
\label{eq:sigma}
\end{equation}
where ${\cal P}$ denotes the principal part.

\subsection{Energy spectrum}
\label{sec:spectrum}

The spectral function of the open quantum system is given by\cite{Stefanucci13}
\begin{equation}
A(\omega) = \frac{i}{2\pi}\left[G(\omega)-G^\dagger(\omega)\right],
\label{eq:A_def}
\end{equation}
and may be decomposed as 
\begin{equation}
A(\omega) =\sum_{\alpha=1}^M A_\alpha(\omega) + \sum_\ell |\ell\rangle \langle \ell| \delta(\omega-\omega_\ell), 
\label{eq:A_result}
\end{equation}
where 
\begin{equation}
A_\alpha(\omega) = \frac{1}{2\pi} G(\omega) \Gamma^\alpha(\omega) G^\dagger(\omega)
\label{eq:A_alpha}
\end{equation}
is the partial spectral function due to scattering states incident on the system from reservoir $\alpha$, and the sum over $\ell$ includes any
{\em localized states} that are not broadened due to the coupling with the reservoir(s).  The localized states, if any, satisfy
\begin{equation}
\left[\omega_\ell - H_{\rm sys} -\sigma(\omega_\ell)+i\Gamma(\omega_\ell)/2 \right] |\ell \rangle  =  0.
\label{eq:loc}
\end{equation}
and their energies are denoted by $\omega_\ell$.
A similar condition was recently analyzed in the context of destructive quantum interferences.\cite{Reuter14,Sam_ang_2017}

The density of states of the open quantum system is 
\begin{equation}
g(\omega)=\Tr \{A(\omega)\}= g_{\rm reg}(\omega)
+\sum_\ell \delta(\omega-\omega_\ell), 
\end{equation}
where 
\begin{equation}
g_{\rm reg}(\omega)=\sum_{\alpha=1}^M g_\alpha(\omega)
\label{eq:g_reg}
\end{equation}
is the non-singular part of the spectrum, and\cite{Gasparian96,Gramespacher97,Stafford16} 
\begin{equation}
g_\alpha(\omega) = \Tr \{A_\alpha(\omega)\}
\label{eq:g_alpha}
\end{equation}
is the partial density of states of the system due to scattering states incident on the system from reservoir $\alpha$.
It should be emphasized that an open quantum system is a subsystem of a larger system, and $g(\omega)$ gives the energy spectrum of
the whole system projected onto the Hilbert space of the subsystem.
Other prescriptions for partitioning into subsystem and environment are also possible.\cite{Friedel1958,Bruch2016}

Similarly, the local density of states is 
given by
\begin{eqnarray}
g(\omega;{\bf x}) & = & \langle {\bf x} | A(\omega) | {\bf x}\rangle
\nonumber \\
& = & \sum_{\alpha=1}^M g_\alpha(\omega;{\bf x}) + \sum_\ell |\psi_\ell({\bf x})|^2 \delta(\omega-\omega_\ell),
\label{eq:g_x} 
\end{eqnarray}
where the local partial density of states associated with reservoir $\alpha$ is\cite{Gasparian96,Gramespacher97,Stafford16}
\begin{equation}
g_\alpha(\omega;{\bf x})  =  \langle {\bf x} | A_\alpha(\omega) | {\bf x}\rangle.
\label{eq:g_alpha_x}
\end{equation}

\subsection{Thermodynamics}
\label{sec:thermo}

The grand canonical potential $\Omega=E-TS-\mu N$ of an open quantum system in thermodynamic equilibrium at absolute temperature $T$
and chemical potential $\mu$ is given by\cite{Stafford97a}
\begin{equation}
\Omega(\mu,T) = -k_B T \int d\omega \, g(\omega) \ln \left[1+e^{-\beta\left(\omega-\mu\right)}\right],
\label{eq:Omega}
\end{equation}
where $k_B$ is Boltzmann's constant and $\beta=1/k_B T$.  The entropy of the system is given by
\begin{eqnarray}
S&=&-\left.\frac{\partial \Omega}{\partial T}\right|_\mu
= \int d\omega \, g(\omega) s(f(\omega)) \nonumber \\
&=& \int d\omega \, g_{\rm reg}(\omega) s(f(\omega)) +\sum_\ell s(f(\omega_\ell)),
\label{eq:S_eq}
\end{eqnarray}
where
\begin{equation}
s(f)=-k_B [f\ln f + (1-f)\ln(1-f)]
\label{eq:s}
\end{equation}
and 
\begin{equation}
f(\omega) = [e^{\beta(\omega-\mu)}+1]^{-1}
\label{eq:FD}
\end{equation}
is the equilibrium Fermi-Dirac distribution of the reservoir(s).
Similarly, one can define the local entropy density\cite{Stafford17}
\begin{eqnarray}
S({\bf x}) &=& \int d\omega \, g(\omega; {\bf x}) s(f(\omega)) 
\label{eq:S_eq_x} \\
&=& \sum_{\alpha=1}^M \int d\omega \, g_\alpha(\omega; {\bf x}) s(f(\omega))  + \sum_\ell |\psi_\ell({\bf x})|^2 s(f(\omega_\ell)),
\nonumber
\end{eqnarray}
which satisfies $S=\int_{\rm sys} d^3 x \, S({\bf x})$.

\subsection{Nonequilibrium steady states}
\label{sec:theory_noneq}

The equilibrium entropy formulas (\ref{eq:S_eq}) and (\ref{eq:S_eq_x}) can be generalized to the case of an open quantum system in a nonequilibrium
steady state.\cite{Shastry18,Shastry19}  Succinctly, the nonequilibrium steady state of a quantum scattering problem for independent quantum particles can be
decomposed into independent contributions from the scattering states incident from each reservoir, which coexist in real space, but are orthogonal in 
Hilbert space.
The nonequilibrium entropy of the open quantum system is\cite{Shastry18,Shastry19}
\begin{equation}
S= \sum_{\alpha=1}^M \int d\omega \, g_\alpha(\omega) s(f_\alpha(\omega)) + \sum_\ell s(f_\ell),
\label{eq:S_noneq}
\end{equation}
where $f_\alpha(\omega) = [e^{\beta_\alpha(\omega-\mu_\alpha)}+1]^{-1}$ is the Fermi-Dirac distribution of reservoir $\alpha$ and $f_\ell$ is the
occupancy of the $\ell$th localized state.
The local nonequilibrium entropy density is\cite{Shastry18,Shastry19}
\begin{equation}
S({\bf x}) = \sum_{\alpha=1}^M \int d\omega \, g_\alpha(\omega; {\bf x}) s(f_\alpha(\omega)) + \sum_\ell |\psi_\ell({\bf x})|^2 s(f_\ell),
\label{eq:S_noneq_x}
\end{equation}
which satisfies $S=\int_{\rm sys} d^3 x \, S({\bf x})$, with $S$ the nonequilibrium entropy given by Eq.\ (\ref{eq:S_noneq}).

\section{3rd Law for equilibrium systems}
\label{sec:equilibrium}

\begin{theorem}[3rd law of thermodynamics]
For an open quantum system with a finite-dimensional Hilbert space,
\begin{equation}
\lim_{T\rightarrow 0} S(\mu,T)=0 
\label{eq:3rd_law}
\end{equation}
almost everywhere for $\mu\in \mathcal{R}$.
\label{thm:3rd_law}
\end{theorem}

\begin{proof}
We consider the first term on the rhs of Eq.\ (\ref{eq:S_eq}):
\begin{equation}
\lim_{T\to 0}S_{\rm{reg}}(\mu,T) =\lim_{T\to0}\int_{-\infty}^{\infty}\, d\omega g_{\rm{reg}}(\omega) s(f(\omega)),
\label{eq:S_eq_lim}
\end{equation}
and note that 
\begin{equation}
\int_{-\infty}^{\infty}\, d\omega\, g_{\rm{reg}}(\omega) = N_{\rm{reg}} \leq N_{\cal H},
\label{eq:g_reg_integral}
\end{equation}
where $N_{\cal H}=\dim\{\cal H\}$ is the dimension of the Hilbert space ${\cal H}$ of the system. 
The Fermi function $\lim_{T\to0}f(\omega)\rightarrow1- \Theta(\omega-\mu)$, where $\Theta$
is the Heaviside step function.
Therefore, 
\begin{equation}
\begin{aligned}
\lim_{T\to 0}S_{\rm{reg}}(\mu,T) =&\lim_{f\to1}\int_{-\infty}^{\mu}\, d\omega\, g_{\rm{reg}}(\omega) s(f)\\
&+ \lim_{f\to0}\int_{\mu}^{\infty}\, d\omega\, g_{\rm{reg}}(\omega) s(f)\\
=&\, 0,
\label{eq:S_eq_lim2}
\end{aligned}
\end{equation}
since $\lim_{f\to1}s(f)=\lim_{f\to0}s(f)=0$ and the integral of $g(\omega)$ is bounded by the dimension of the Hilbert space 
[Eq.\ \ref{eq:g_reg_integral}].
A similar result, restricted to the {\em resonant level model}, was derived in Ref.\ \onlinecite{Bruch2016}.

The second term from Eq.\ (\ref{eq:S_eq}) has the entropy contribution from the localized states which vanish as $T\to0$ when 
$\mu\neq\omega_{\ell}$ since
\begin{equation}
\lim_{f\to0}s(f(\omega_{\ell}))= \lim_{f\to1}s(f(\omega_{\ell}))=0,
\end{equation}
and when $\mu=\omega_{\ell}$ we get
\begin{equation}
S_{\rm{loc}}=\lim_{T\to0}s(f(\mu=\omega_{\ell}))= k_{B}\log(2),
\end{equation}
and if there are multiple localized states at $\omega=\omega_{\ell}$, we denote the degeneracy as $g_{\ell}$ and may write
\begin{equation}
S_{\rm{loc}}=g_{\ell}k_{B}\log(2).
\end{equation}
The points $\mu=\omega_{\ell}$ have zero measure for $\mu\in\mathcal{R}$ and this completes the proof.

\end{proof}

Theorem \ref{thm:3rd_law} constitutes the general form of the {\em third law of thermodynamics} for open quantum systems, and is the central 
result of this paper. 
Figure \ref{fig:3rd_law} illustrates the behavior of $S(\mu,T)$ as $T\rightarrow 0$ for a model open quantum system consisting of a benzene
ring coupled to an electron reservoir.  

\begin{figure}[htb]
        \includegraphics[width=8.0cm]{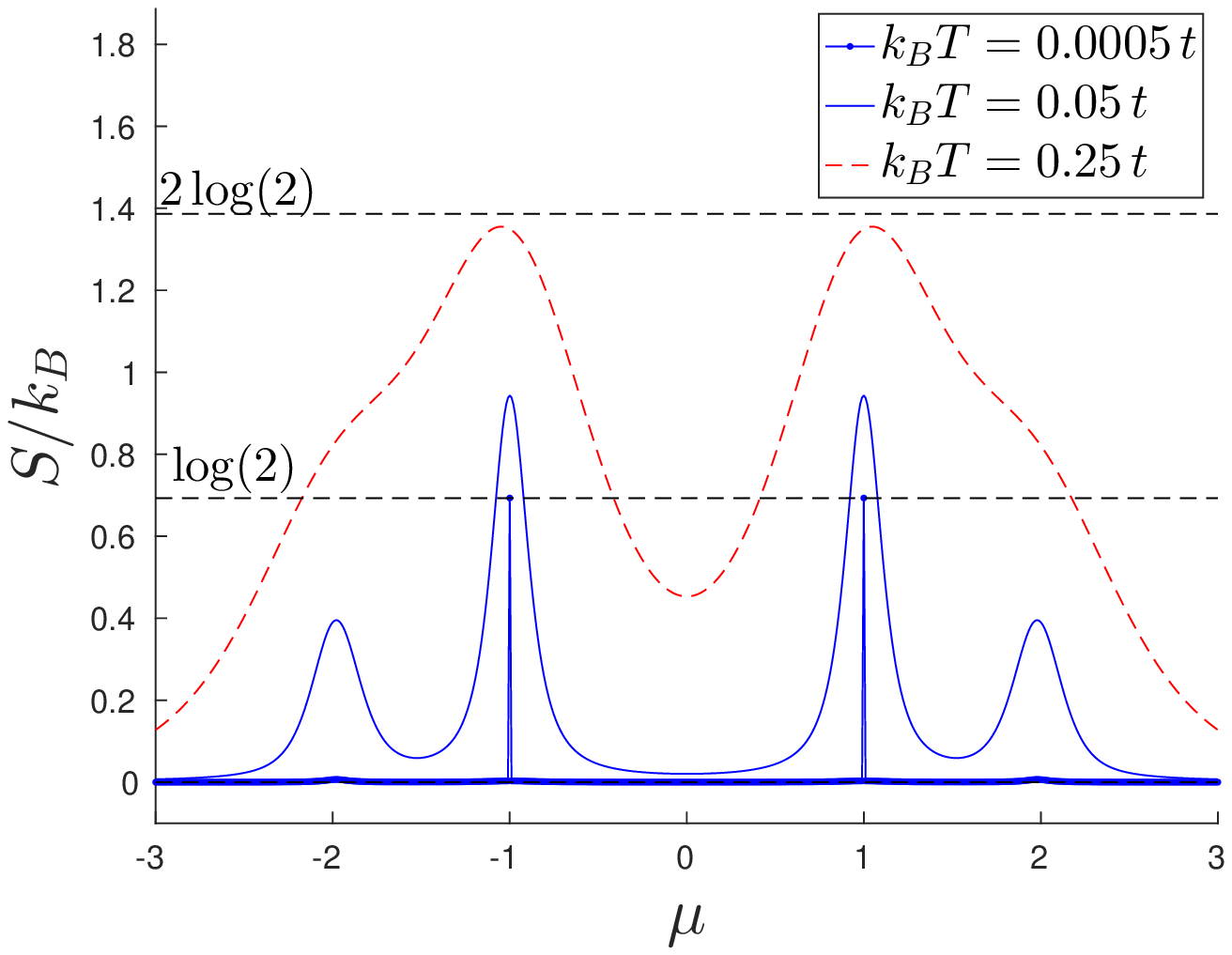}
        \includegraphics[width=8.0cm]{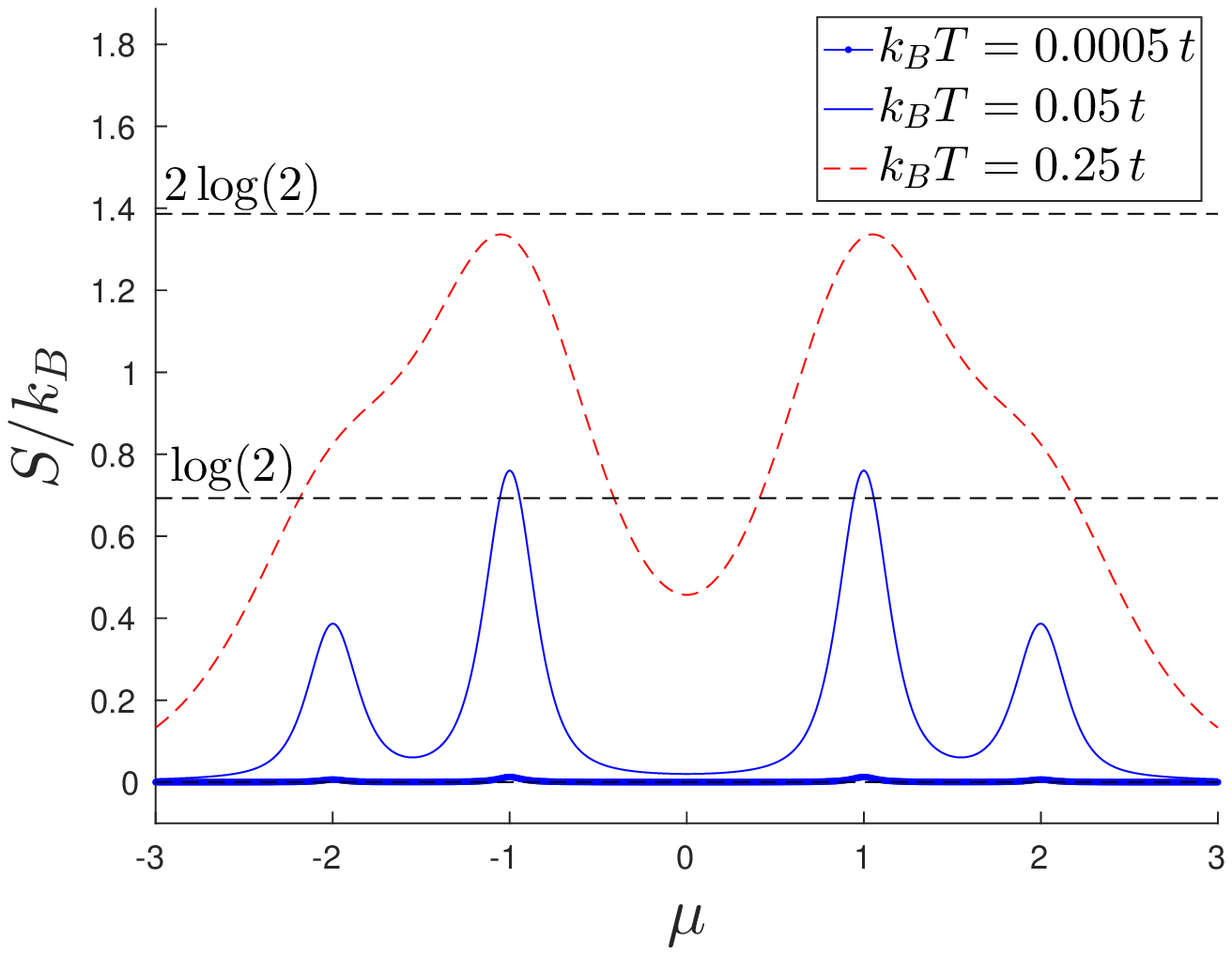}
\caption{
Entropy $S$ of an open quantum system consisting of a benzene ring coupled to an equilibrium electron reservoir, 
plotted as a function of the chemical potential $\mu$ of the reservoir for several temperatures.  The molecule is modelled using H\"uckel theory 
(tight-binding approximation, see Appendix \ref{sec:benzene}) and energies are expressed in units of the nearest-neighbor hybridization $t$. 
Top panel: Generic open system with $\Gamma_{11}=t$ and all other matrix elements of $\Gamma$ zero, illustrating the effect of the localized states
at $\mu/t=\pm 1$. Bottom panel: Fully open quantum system with $\Gamma=(t/6)\mathbb{1}$, illustrating the strong form of the third law, $S(T)\rightarrow 0$ as
$T\rightarrow 0$ $\forall \mu$.
}
\label{fig:3rd_law}
\end{figure}

\subsection{Fully open quantum systems}
\label{sec:fully}

A strong form of the third law of thermodynamics can be shown to hold for {\em fully open} quantum systems.
\begin{definition}
A fully open quantum system is any system for which Eq.\ (\ref{eq:loc}) has no solutions for $\omega_\ell \in {\cal R}$.
\label{def:fully}
\end{definition}
\noindent
In a fully open quantum system, all of the energy eigenstates of the system are broadened due to coupling to the reservoir; there are no localized states.

\begin{corollary}
For a fully open quantum system with a finite-dimensional Hilbert space,
\begin{equation}
\lim_{T\rightarrow 0} S(\mu,T)=0 \ \,  \forall \, \mu.
\label{eq:3rd_law_fully}
\end{equation}
\label{cor:3rd_law_fully}
\end{corollary}

\begin{proof}

For a fully open system, Eq.\ (\ref{eq:g_reg_integral}) holds with $N_{\rm{reg}}=N_{\cal H}$ and there are
no localized states. 
Proof follows directly from Eq.\ (\ref{eq:S_eq_lim2}) and theorem \ref{thm:3rd_law} holds $\forall\, \mu\in\mathcal{R}$.
\end{proof}

Since $S$ is a smooth function of $\mu\in\mathcal{R}$ and $T\in(0,\infty)$,
we provide the leading-order low-temperature expansion:\cite{Shastry15}
\begin{equation}
S(\mu,T) 
=\frac{\pi^2}{3} g(\mu) k_B^2 T.
\label{eq:S_fully}
\end{equation}
$S$ scales as a higher power of $T$ as $T\rightarrow 0$ if $g(\mu)=0$.

\subsubsection{Sufficient condition for a fully open quantum system}
\label{sec:suffic_fully}

A sufficient condition for a fully open quantum system is
\begin{equation}
\Gamma(\omega) |\psi\rangle \neq 0 \;\;\; \forall \, |\psi\rangle \in {\cal H},
\label{eq:fully1}
\end{equation}
where ${\cal H}$ denotes the Hilbert space of $H_{\rm sys}$.
Eq.\ (\ref{eq:fully1}) holds if $\rank\{\Gamma\} = \dim{\cal H}$ since $\Gamma \geq 0$.\cite{Shastry16}
Figure \ref{fig:3rd_law}(lower panel) illustrates the case of a fully open quantum system, for the same model of a benzene ring coupled to a reservoir,
but with $\Gamma=\gamma \mathbb{1}$, so that $\rank\{\Gamma\}=6$.  Note the absence of any localized states in the entropy spectrum.

\subsubsection{Example}
\label{sec:example_fully}

Eq.\ (\ref{eq:fully1}) is a sufficient condition for a fully open quantum system, but is not a necessary condition.  To see that this is the case, 
consider the following example.  Let the system-reservoir coupling be in the broad-band limit so that $\sigma(\omega)=0$ and $\Gamma(\omega)=\mbox{const.}$
Let $|\nu\rangle$ be an eigenstate of $H_{\rm sys}$, $H_{\rm sys}|\nu\rangle=\varepsilon_\nu |\nu\rangle$, and suppose
\begin{equation}
\langle \nu|\Gamma|\nu\rangle \neq 0 \;\;\; \forall \, \nu.
\label{eq:example1}
\end{equation}
Then it is straightforward to show that there exist no solutions to Eq.\ (\ref{eq:loc}) for real energies, provided the eigenenergies $\varepsilon_\nu$
of $H_{\rm sys}$ are nondegenerate.  The following rank-1 $\Gamma$ matrix is an example leading to a fully open quantum system with a nondegenerate
spectrum:
\begin{equation}
\Gamma=\gamma|\gamma\rangle\langle \gamma|, \; \mbox{where} \; |\gamma\rangle=\sum_\nu C_\nu |\nu\rangle,
\label{eq:Gamma_example}
\end{equation}
and $C_\nu \neq 0$ are complex numbers.

However, if the energy eigenvalue $\varepsilon_\nu$ is degenerate and the support of $\Gamma$ does not span the degenerate subspace, 
then there will be at least one localized state of energy $\varepsilon_\nu$.  This case is illustrated in Fig.\ \ref{fig:3rd_law}(upper panel).

\section{3rd Law for Nonequilibrium systems}
\label{sec:noneq}

For nonequilibrium systems, one also needs the Keldysh ``lesser'' Green's function\cite{Stefanucci13,Stafford17}
\begin{equation}
G^<_{ij}(t)=i\langle d^\dagger_j(0) d_i(t)\rangle,
\label{eq:G<_def}
\end{equation}
which determines the occupancies of the states out of equilibrium.  Its Fourier transform obeys the Keldysh equation\cite{Stefanucci13}
\begin{equation}
G^<(\omega) = G(\omega)\Sigma^<(\omega) G^\dagger(\omega),
\label{eq:Keldysh}
\end{equation}
where the lesser self-energy is given by
\begin{equation}
\Sigma^<(\omega)=i\sum_{\alpha=1}^M \Gamma^\alpha(\omega) f_\alpha(\omega) + 2i\eta \sum_\ell |\ell\rangle \langle \ell| f_\ell,
\label{eq:Sigma<}
\end{equation}
where $f_\alpha(\omega)$ and $f_\ell$ are the Fermi-Dirac distribution of reservoir $\alpha$ and the nonequilibrium occupancy of the $\ell$th
localized state, respectively, as discussed in Sec.\ \ref{sec:theory_noneq}.
Inserting Eq.\ (\ref{eq:Sigma<}) in Eq.\ (\ref{eq:Keldysh}), and using Eq.\ (\ref{eq:A_alpha}), one finds
\begin{equation}
G^<(\omega)=2\pi i \left[\sum_{\alpha=1}^M A_\alpha(\omega)f_\alpha(\omega) + \sum_\ell f_\ell |\ell\rangle \langle \ell| 
\delta(\omega-\omega_\ell)\right].
\label{eq:G<_result}
\end{equation}

The particle density $N({\bf x})$ and energy density $E({\bf x})$ of the nonequilibrium quantum system are given by\cite{Stafford17}
\begin{equation}
N({\bf x}) 
= \int d\omega \, g(\omega; {\bf x}) f(\omega,{\bf x}),
\label{eq:Nofx}
\end{equation}
\begin{equation}
E({\bf x}) 
= \int d\omega \, g(\omega; {\bf x}) \omega f(\omega,{\bf x}),
\label{eq:Eofx}
\end{equation}
where 
\begin{eqnarray}
f(\omega;{\bf x}) & \equiv & \frac{\langle {\bf x} | G^<(\omega)|{\bf x}\rangle}{2\pi i g(\omega;{\bf x})}
\label{eq:f_noneq} \\
& = & \left\{\begin{array}{lr} f_\ell, & \omega=\omega_\ell \\ 
\frac{1}{g(\omega;{\bf x})}\sum_\alpha g_\alpha(\omega; {\bf x}) f_\alpha(\omega), & \omega \neq \omega_\ell \end{array}\right.
\nonumber
\end{eqnarray}
is the local nonequilibrium distribution function of the system.\cite{Stafford17,Ness14}
Note that $f(\omega;{\bf x})$ may be discontinuous at $\omega=\omega_\ell$.
Eq.\ (\ref{eq:Nofx}) holds quite generally, while Eq.\ (\ref{eq:Eofx}) holds for the case of independent fermions.

\subsection{Local thermodynamic variables $T({\bf x})$, $\mu({\bf x})$}
\label{eq:local_var}

In order to address the applicability of the 3rd law of thermodynamics in systems out of equilibrium, one needs 
a concept of local temperature.\cite{DiVentra09,Jacquet2012,Bergfield2013demon,Meair14,Shastry15,Bergfield15,Stafford16,Shastry16,Stafford17}
Here, we define the local temperature $T({\bf x})$ and chemical potential $\mu({\bf x})$ as those measured by 
a floating broad-band thermoelectric probe coupled weakly to the system at the point ${\bf x}$ 
(see Refs.\ \onlinecite{Stafford16,Shastry16,Stafford17} for discussion).
The floating probe consists of an equilibrium electron reservoir whose temperature and chemical potential are adjusted such that the net flow of
charge and heat into the probe is zero.  Under these conditions, the Fermi-Dirac distribution of the probe reservoir $f_p(\omega;{\bf x})$ defines
$T({\bf x})$ and $\mu({\bf x})$ via
\begin{equation}
f_p(\omega;{\bf x}) = \left(e^{\beta({\bf x})[\omega - \mu({\bf x)}]}+1\right)^{-1}.
\label{eq:f_p}
\end{equation}
$f_p(\omega;{\bf x})$ is the unique\cite{Shastry16} {\em equilibrium distribution} that reproduces the particle and energy densities of
the nonequilbrium system given by Eqs.\ (\ref{eq:Nofx}) and (\ref{eq:Eofx}), respectively.\cite{Stafford16,Stafford17}

\begin{definition}
Let $S_p({\bf x})$ be the entropy density 
of the fictitious local equilibrium state obtained by inserting the probe distribution function
$f_p(\omega;{\bf x})$, Eq.\ (\ref{eq:f_p}), into the equilibrium formula, Eq.\ (\ref{eq:S_eq_x}).  
\label{def:S_p}
\end{definition}
\noindent
$S_p({\bf x})$ so defined clearly satisfies Theorem \ref{thm:3rd_law}.

\begin{definition}
Let $S_s({\bf x})$ be a measure of the local nonequilibrium entropy density of the system defined by inserting the local nonequilibrium distribution function
$f(\omega;{\bf x})$, Eq.\ (\ref{eq:f_noneq}), directly into the equilibrium formula, Eq.\ (\ref{eq:S_eq_x}).
\label{def:S_s}
\end{definition}
\noindent
$S_s({\bf x})$ so defined 
was introduced previously in Ref.\ \onlinecite{Stafford17}, and 
may be thought of as the local entropy density
inferred by an observer with strictly local knowledge of the nonequilibrium state of the system.\cite{Shastry18,Shastry19}
\\

\noindent
{\bf Maximum entropy principle.} 
$S_s({\bf x}) \leq S_p({\bf x}).$
\\

\noindent
The fictitious local equilibrium distribution $f_p(\omega;{\bf x})$ satisfies the same constraints, Eqs.\ (\ref{eq:Nofx}) and (\ref{eq:Eofx}), as the actual
nonequilibrium distribution $f(\omega;{\bf x})$. 
Therefore, by the {\em maximum entropy principle}, $S_p({\bf x}) \geq S_s({\bf x})$.
For a detailed variational argument for the case of a fully open quantum system, see Ref.\ \onlinecite{Stafford17}.

\begin{theorem}[3rd law for nonequilibrium systems]
For an open quantum system in a nonequilibrium steady state, and with a finite-dimensional Hilbert space,
\begin{equation}
S({\bf x}) \rightarrow 0 \; \mbox{as} \;\, T({\bf x}) \rightarrow 0
\label{eq:3rd_law_noneq}
\end{equation}
almost everywhere in space.
\end{theorem}

\begin{proof}
$S({\bf x})$ is defined by Eq.\ (\ref{eq:S_noneq_x}).
The contributions of any localized states to $S({\bf x})$ and $S_s({\bf x})$ are identical.  
Therefore, any difference between $S({\bf x})$ and $S_s({\bf x})$ is due to scattering states.
From the concavity of the function $s(f)$ defined by Eq.\ (\ref{eq:s}), it follows that 
$S_s({\bf x}) \geq S({\bf x})$, because
\begin{equation}
s(f)\geq \sum_{\alpha=1}^M \lambda_\alpha s(f_\alpha),
\label{eq:concave}
\end{equation}
where
\begin{equation}
f(\omega;{\bf x}) = \sum_{\alpha=1}^M \lambda_\alpha(\omega;{\bf x}) f_\alpha(\omega), \;\;\; \sum_{\alpha=1}^M \lambda_\alpha =1,
\label{eq:fofx2}
\end{equation}
and $\lambda_\alpha(\omega;{\bf x}) =g_\alpha(\omega;{\bf x})/g(\omega;{\bf x}) \geq 0$.
But $S_s({\bf x})\leq S_p({\bf x})$ by the maximum entropy principle. 
Therefore, $S({\bf x})\leq S_p({\bf x})$.
\end{proof}
Finally, the limiting behavior of $S_p({\bf x})$ can be obtained from the Sommerfeld expansion of Eq.\ (\ref{eq:S_eq_x}),
\begin{eqnarray}
S_p({\bf x}) & \stackrel{T({\bf x})\rightarrow 0}{\sim} & 
\frac{\pi^2}{3} 
k_B^2 T({\bf x})
\sum_{\alpha=1}^M g_{\alpha}(\mu({\bf x}); {\bf x}) \nonumber \\
& & +\sum_\ell |\psi_\ell({\bf x})|^2 s(f_p(\omega_\ell;{\bf x})),
\label{eq:S_p_lim}
\end{eqnarray}
provided $\sum_{\alpha} g_{\alpha}(\mu({\bf x}); {\bf x})\neq 0$, and if it is zero, the first term on the r.h.s.\ vanishes as a higher power of
$T({\bf x})$.
Therefore,
\begin{equation}
\lim_{T({\bf x})\rightarrow 0} S({\bf x}) = \left\{ \begin{array}{lr} 0, & \mu({\bf x}) \neq \omega_\ell, \\
k_B \ln 2 \sum_\ell |\psi_\ell({\bf x})|^2, & \mu({\bf x}) = \omega_\ell. \end{array} \right.
\label{eq:S_noneq_limit}
\end{equation}
In a nonequilibrium system, where $\mu({\bf x})$ varies as a continuous function of position,\cite{Buttiker89,Bergfield14}
the contribution from localized states in Eq.\ (\ref{eq:S_noneq_limit}) vanishes everywhere in space except on a set of measure zero.

\section{Conclusions}
\label{sec:conclude}

The third law of thermodynamics (Nernst's heat theorem) is shown to hold true for open quantum systems, both in and out of equilibrium.  
This is proven for a generic open quantum system of independent fermions with strong coupling to reservoirs but no 
time-dependent external driving.

Our analysis of the third law was shown to hold for quite general partitioning of Hilbert space into an open subsystem and an environment (reservoir(s)).
It is important to emphasize that such a partitioning is well defined in Hilbert space, but that physical observables relevant for thermodynamics 
may involve both system and reservoir degrees of freedom.\cite{Bruch2016}

The laws of thermodynamics are arguably the most fundamental of all physical principles, more general than any particular theory, such as quantum mechanics
or general relativity. For this reason,
claims of violations of the laws of thermodynamics in open and/or nonequilibrium quantum systems
should be treated with appropriate skepticism, and subjected to the most stringent examination before they are taken at face value.

\begin{acknowledgments}
This work was
supported by the U.S.\ Department of Energy
(DOE), Office of Science under Award No.\ DE-SC0006699.
\end{acknowledgments}

\appendix

\section{Benzene molecular junction}
\label{sec:benzene}

As a specific example of an open quantum system, a molecular junction consisting of a single benzene molecule
coupled to a macroscopic electron reservoir 
is analyzed in Fig.\ \ref{fig:3rd_law}.  The molecule is modelled using H\"uckel theory, with
Hamiltonian
\begin{equation}
H_{\rm sys} = -\sum_{j=1}^6 \left(t e^{i\pi\Phi/3\phi_0} d^\dagger_{j} d_{j+1} + \mbox{H.c.}\right),
\label{eq:H_benzene}
\end{equation}
where the benzene ring is threaded by a magnetic flux $\Phi$, and
$\phi_0=hc/e$ is the magnetic flux quantum.
The system is depicted in Fig.\ \ref{fig:benzene}.
For benzene, the nearest-neighbor coupling\cite{Barr12} is taken as $t=2.7\mbox{eV}$ and the fermion
annihilation operators satisfy periodic boundary conditions, $d_0=d_6$, $d_7=d_1$.
For simplicity, spin is neglected, which reduces the total entropy of the system by a factor of two.

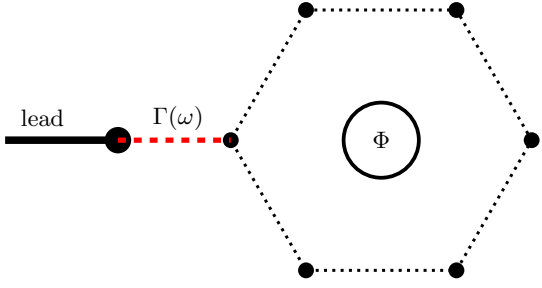
\begin{figure}
\centering
\begin{tikzpicture}
\fill (-1,1.732) circle  (3pt);
\fill (1,1.732) circle  (3pt);
\fill (2,0) circle  (3pt);
\fill (1,-1.732) circle  (3pt);
\fill (-1,-1.732) circle  (3pt);
\fill (-2,0) circle  (3pt);
\draw[dotted, line width=0.4mm]  (-1,1.732)--(1,1.732)--(2,0)--(1,-1.732)--(-1,-1.732)--(-2,0)--(-1,1.732);
\draw[line width=0.5mm] (0,0) circle (0.5cm);
\node (0,0) {$\Phi$};
\draw[line width=1mm] (-5,0)--(-3.5,0);
\fill (-3.5,0) circle (5pt);
\draw[dashed, line width=0.7mm, red] (-3.5,0)--(-2,0);
\node at (-4.5,0.3) {lead};
\node at (-2.7,0.3) {$\Gamma(\omega)$};
\end{tikzpicture}
\caption{An open quantum system consisting of a benzene molecule threaded by a magnetic flux $\Phi$, coupled to a single macroscopic reservoir (lead).}
\label{fig:benzene}
\end{figure}

\subsection{Fully open quantum system}
\label{sec:appendix_fully}

Let us first consider the case of a fully open quantum system with retarded self-energy
\begin{equation}
\Sigma_{nm}(\omega) = -\frac{i \gamma}{2} \delta_{nm}.
\label{eq:Sigma_fully}
\end{equation}
For this case, the density of states may be calculated in closed form as
\begin{equation}
    g(\omega,\Phi)=\frac{6}{\pi} \frac{K(\omega)Q(\omega)-L(\omega)P(\omega,\Phi)}{P(\omega,\Phi)^2+Q(\omega)^2},
\label{eq:DOS_fully}
\end{equation}
where the four functions $K,L,P$ and $Q$ are given by
\begin{widetext}
\begin{equation}
    K(\omega)=x(x^2-3y^2-t^2)
(x^2-y^2-3t^2)-2xy^2(3x^2-y^2-t^2),
\end{equation}
\begin{equation}
    L(\omega)=  2x^2y(x^2-3y^2-t^2)+y(3x^2-y^2-t^2)(x^2-y^2-3t^2),
\end{equation}
\begin{equation}
     P(\omega,\Phi)=x^2(x^2-3y^2-3t^2)^2-y^2(3x^2-y^2-3t^2)^2-4t^6\cos^2(\varphi)
\end{equation}
and
\begin{equation}
    Q(\omega) = 2xy(x^2-3y^2-3t^2)
(3x^2-y^2-3t^2),
\end{equation}
\end{widetext}
respectively,
where $x \equiv \omega$, $y \equiv \gamma/2$, and $\varphi\equiv\pi \Phi/\phi_0$.
It can be shown that $P(\omega,\Phi)^2+Q(\omega)^2 \neq 0$ $\forall \omega \in {\mathcal R}$,
so $g(\omega,\Phi)$ is well-defined everywhere along the real-$\omega$ axis.  This case is illustrated in Fig.\ \ref{fig:3rd_law}(lower panel).

\subsection{Open system with localized states}
\label{sec:appendix_localized}

Next, let us consider coupling the benzene molecule to the reservoir via a single covalent bond, as shown in Fig.\ \ref{fig:benzene}.
The retarded self-energy for this case is
\begin{equation}
\Sigma_{nm}(\omega) = -\frac{i\gamma}{2} \delta_{m1}\delta_{n1}.
\label{eq:Sigma_loc}
\end{equation}
Since $\rank\{\Gamma\}=1$, localized states occur whenever there is a degeneracy in the spectrum of $H_{\rm sys}$, as discussed in Sec.\
\ref{sec:example_fully}.
The density of states is given by
\begin{widetext}
\begin{equation}
   g(\omega,\Phi)=-\frac{1}{\pi t} \frac{-UE^2(E^2-3)^2(E^4+3)-4U(5E^4-12E^2+3)\cos^2(\varphi)}{
\left[4\cos^2(\varphi)-E^2(E^2-3)^2\right]^2+\left[UE(E^2-3)(E^2-1)\right]^2},
\label{eq:DOS_loc}
\end{equation}
\end{widetext}
where $E\equiv \omega/t$, $U\equiv \gamma/2t$, and again $\varphi\equiv\pi \Phi/\phi_0$.
A careful analysis of the denominator in Eq.\ (\ref{eq:DOS_loc}) reveals that
$g(\omega,\Phi)$ is singular for $(E,\Phi/\phi_0)= (\pm 1,n)$, $(E,\Phi/\phi_0)= (\pm \sqrt{3},n+1/2)$, and $(E,\Phi/\phi_0)= (0, n+1/2)$,
where $n$ is an integer.
We treat the two special values of the flux separately in the next two sections.

\subsubsection{Integer flux: $\Phi=n\phi_0$}
\label{sec:integer_flux}

For $\Phi=n\phi_0$, $g(E)$ is regular except at $E=\pm 1$.  Evaluating the limiting behavior as $\Phi\rightarrow n\phi_0$, one finds
\begin{widetext}
\begin{equation}
  g(\omega,\Phi=n\phi_0)=\frac{\delta(E-1)}{t}+\frac{\delta(E+1)}{t}+\frac{U}{\pi t} \frac{E^6-4E^4+3E^2+12}{(E^2-1)^2(E^2-4)^2+U^2E^2(E^2-3)^2}.
\end{equation}
\end{widetext}
For this case,
\begin{equation}
\lim_{T\rightarrow 0}S(\mu,T)=\left\{ \begin{array}{lr} k_B \ln 2, & \mu=\pm t, \\ 0, & \mbox{otherwise.} \end{array}\right.
\end{equation}
The effect of the localized states at $E=\pm 1$ on the entropy is illustrated in Fig.\ \ref{fig:3rd_law}(upper panel).

\subsubsection{Half-odd integer flux: $\Phi=(n+1/2)\phi_0$}
\label{sec:half_integer_flux}

For $\Phi=(n+1/2)\phi_0$, $g(E)$ is regular except at $E=0,\pm \sqrt{3}$.  Evaluating the limiting behavior as $\Phi\rightarrow (n+1/2)\phi_0$, one finds
\begin{widetext}
\begin{equation}
 g\left(\omega,(n+1/2)\phi_0\right)
=\frac{\delta(E)}{t}+\frac{\delta(E-\sqrt{3})}{t}+\frac{\delta(E+\sqrt{3})}{t}+\frac{U}{\pi t} \frac{E^4+3}{E^2(E^2-3)^2+U^2(E^2-1)^2}.
\end{equation}
\end{widetext}
For this case,
\begin{equation}
\lim_{T\rightarrow 0}S(\mu,T)=\left\{ \begin{array}{lr} k_B \ln 2, & \mu=\pm \sqrt{3}t, \\ 
k_B \ln 2, & \mu=0,\\
0, & \mbox{otherwise.} \end{array}\right.
\end{equation}

\bibliography{refs,3rdLawLit} 

\begin{thebibliography}{37}%
\makeatletter
\providecommand \@ifxundefined [1]{%
 \@ifx{#1\undefined}
}%
\providecommand \@ifnum [1]{%
 \ifnum #1\expandafter \@firstoftwo
 \else \expandafter \@secondoftwo
 \fi
}%
\providecommand \@ifx [1]{%
 \ifx #1\expandafter \@firstoftwo
 \else \expandafter \@secondoftwo
 \fi
}%
\providecommand \natexlab [1]{#1}%
\providecommand \enquote  [1]{``#1''}%
\providecommand \bibnamefont  [1]{#1}%
\providecommand \bibfnamefont [1]{#1}%
\providecommand \citenamefont [1]{#1}%
\providecommand \href@noop [0]{\@secondoftwo}%
\providecommand \href [0]{\begingroup \@sanitize@url \@href}%
\providecommand \@href[1]{\@@startlink{#1}\@@href}%
\providecommand \@@href[1]{\endgroup#1\@@endlink}%
\providecommand \@sanitize@url [0]{\catcode `\\12\catcode `\$12\catcode
  `\&12\catcode `\#12\catcode `\^12\catcode `\_12\catcode `\%12\relax}%
\providecommand \@@startlink[1]{}%
\providecommand \@@endlink[0]{}%
\providecommand \url  [0]{\begingroup\@sanitize@url \@url }%
\providecommand \@url [1]{\endgroup\@href {#1}{\urlprefix }}%
\providecommand \urlprefix  [0]{URL }%
\providecommand \Eprint [0]{\href }%
\providecommand \doibase [0]{http://dx.doi.org/}%
\providecommand \selectlanguage [0]{\@gobble}%
\providecommand \bibinfo  [0]{\@secondoftwo}%
\providecommand \bibfield  [0]{\@secondoftwo}%
\providecommand \translation [1]{[#1]}%
\providecommand \BibitemOpen [0]{}%
\providecommand \bibitemStop [0]{}%
\providecommand \bibitemNoStop [0]{.\EOS\space}%
\providecommand \EOS [0]{\spacefactor3000\relax}%
\providecommand \BibitemShut  [1]{\csname bibitem#1\endcsname}%
\let\auto@bib@innerbib\@empty
\bibitem [{\citenamefont {Levy}, \citenamefont {Alicki},\ and\ \citenamefont
  {Kosloff}(2012{\natexlab{a}})}]{Levy2012}%
  \BibitemOpen
  \bibfield  {author} {\bibinfo {author} {\bibfnamefont {A.}~\bibnamefont
  {Levy}}, \bibinfo {author} {\bibfnamefont {R.}~\bibnamefont {Alicki}}, \ and\
  \bibinfo {author} {\bibfnamefont {R.}~\bibnamefont {Kosloff}},\ }\href
  {\doibase 10.1103/PhysRevE.85.061126} {\bibfield  {journal} {\bibinfo
  {journal} {Phys. Rev. E}\ }\textbf {\bibinfo {volume} {85}},\ \bibinfo
  {pages} {061126} (\bibinfo {year} {2012}{\natexlab{a}})}\BibitemShut
  {NoStop}%
\bibitem [{\citenamefont {Kol\'ar}\ \emph {et~al.}(2012)\citenamefont
  {Kol\'ar}, \citenamefont {Gelbwaser-Klimovsky}, \citenamefont {Alicki},\ and\
  \citenamefont {Kurizki}}]{Kolar2012}%
  \BibitemOpen
  \bibfield  {author} {\bibinfo {author} {\bibfnamefont {M.}~\bibnamefont
  {Kol\'ar}}, \bibinfo {author} {\bibfnamefont {D.}~\bibnamefont
  {Gelbwaser-Klimovsky}}, \bibinfo {author} {\bibfnamefont {R.}~\bibnamefont
  {Alicki}}, \ and\ \bibinfo {author} {\bibfnamefont {G.}~\bibnamefont
  {Kurizki}},\ }\href {\doibase 10.1103/PhysRevLett.109.090601} {\bibfield
  {journal} {\bibinfo  {journal} {Phys. Rev. Lett.}\ }\textbf {\bibinfo
  {volume} {109}},\ \bibinfo {pages} {090601} (\bibinfo {year}
  {2012})}\BibitemShut {NoStop}%
\bibitem [{\citenamefont {Cleuren}, \citenamefont {Rutten},\ and\ \citenamefont
  {Van~den Broeck}(2012)}]{Cleuren2012}%
  \BibitemOpen
  \bibfield  {author} {\bibinfo {author} {\bibfnamefont {B.}~\bibnamefont
  {Cleuren}}, \bibinfo {author} {\bibfnamefont {B.}~\bibnamefont {Rutten}}, \
  and\ \bibinfo {author} {\bibfnamefont {C.}~\bibnamefont {Van~den Broeck}},\
  }\href {\doibase 10.1103/PhysRevLett.108.120603} {\bibfield  {journal}
  {\bibinfo  {journal} {Phys. Rev. Lett.}\ }\textbf {\bibinfo {volume} {108}},\
  \bibinfo {pages} {120603} (\bibinfo {year} {2012})}\BibitemShut {NoStop}%
\bibitem [{\citenamefont {Levy}, \citenamefont {Alicki},\ and\ \citenamefont
  {Kosloff}(2012{\natexlab{b}})}]{Levy2012c}%
  \BibitemOpen
  \bibfield  {author} {\bibinfo {author} {\bibfnamefont {A.}~\bibnamefont
  {Levy}}, \bibinfo {author} {\bibfnamefont {R.}~\bibnamefont {Alicki}}, \ and\
  \bibinfo {author} {\bibfnamefont {R.}~\bibnamefont {Kosloff}},\ }\href
  {\doibase 10.1103/PhysRevLett.109.248901} {\bibfield  {journal} {\bibinfo
  {journal} {Phys. Rev. Lett.}\ }\textbf {\bibinfo {volume} {109}},\ \bibinfo
  {pages} {248901} (\bibinfo {year} {2012}{\natexlab{b}})}\BibitemShut
  {NoStop}%
\bibitem [{\citenamefont {Kosloff}(2013)}]{Kosloff2013}%
  \BibitemOpen
  \bibfield  {author} {\bibinfo {author} {\bibfnamefont {R.}~\bibnamefont
  {Kosloff}},\ }\href {\doibase 10.3390/e15062100} {\bibfield  {journal}
  {\bibinfo  {journal} {Entropy}\ }\textbf {\bibinfo {volume} {15}},\ \bibinfo
  {pages} {2100} (\bibinfo {year} {2013})}\BibitemShut {NoStop}%
\bibitem [{\citenamefont {Masanes}\ and\ \citenamefont
  {Oppenheim}(2017)}]{Masanes2017}%
  \BibitemOpen
  \bibfield  {author} {\bibinfo {author} {\bibfnamefont {L.}~\bibnamefont
  {Masanes}}\ and\ \bibinfo {author} {\bibfnamefont {J.}~\bibnamefont
  {Oppenheim}},\ }\href {https://doi.org/10.1038/ncomms14538} {\bibfield
  {journal} {\bibinfo  {journal} {Nature Communications}\ }\textbf {\bibinfo
  {volume} {8}} (\bibinfo {year} {2017})}\BibitemShut {NoStop}%
\bibitem [{\citenamefont {Ford}\ and\ \citenamefont
  {O'Connell}(2005)}]{Ford2005}%
  \BibitemOpen
  \bibfield  {author} {\bibinfo {author} {\bibfnamefont {G.~W.}\ \bibnamefont
  {Ford}}\ and\ \bibinfo {author} {\bibfnamefont {R.~F.}\ \bibnamefont
  {O'Connell}},\ }\href {\doibase 10.1016/j.physe.2005.05.004} {\bibfield
  {journal} {\bibinfo  {journal} {Physica E: Low-dimensional Systems and
  Nanostructures}\ }\textbf {\bibinfo {volume} {29}},\ \bibinfo {pages} {82}
  (\bibinfo {year} {2005})}\BibitemShut {NoStop}%
\bibitem [{\citenamefont {O'Connell}(2006)}]{OConnell2006}%
  \BibitemOpen
  \bibfield  {author} {\bibinfo {author} {\bibfnamefont {R.~F.}\ \bibnamefont
  {O'Connell}},\ }\href {\doibase 10.1007/s10955-006-9151-6} {\bibfield
  {journal} {\bibinfo  {journal} {Journal of Statistical Physics}\ }\textbf
  {\bibinfo {volume} {124}},\ \bibinfo {pages} {15} (\bibinfo {year}
  {2006})}\BibitemShut {NoStop}%
\bibitem [{\citenamefont {Nieuwenhuizen}\ and\ \citenamefont
  {Allahverdyan}(2002)}]{Nieuwenhuizen2002}%
  \BibitemOpen
  \bibfield  {author} {\bibinfo {author} {\bibfnamefont {T.~M.}\ \bibnamefont
  {Nieuwenhuizen}}\ and\ \bibinfo {author} {\bibfnamefont {A.~E.}\ \bibnamefont
  {Allahverdyan}},\ }\href {\doibase 10.1103/PhysRevE.66.036102} {\bibfield
  {journal} {\bibinfo  {journal} {Phys. Rev. E}\ }\textbf {\bibinfo {volume}
  {66}},\ \bibinfo {pages} {036102} (\bibinfo {year} {2002})}\BibitemShut
  {NoStop}%
\bibitem [{\citenamefont {Esposito}, \citenamefont {Ochoa},\ and\ \citenamefont
  {Galperin}(2015)}]{Esposito2015}%
  \BibitemOpen
  \bibfield  {author} {\bibinfo {author} {\bibfnamefont {M.}~\bibnamefont
  {Esposito}}, \bibinfo {author} {\bibfnamefont {M.~A.}\ \bibnamefont {Ochoa}},
  \ and\ \bibinfo {author} {\bibfnamefont {M.}~\bibnamefont {Galperin}},\
  }\href {\doibase 10.1103/PhysRevB.92.235440} {\bibfield  {journal} {\bibinfo
  {journal} {Phys. Rev. B}\ }\textbf {\bibinfo {volume} {92}},\ \bibinfo
  {pages} {235440} (\bibinfo {year} {2015})}\BibitemShut {NoStop}%
\bibitem [{\citenamefont {Ludovico}\ \emph {et~al.}(2014)\citenamefont
  {Ludovico}, \citenamefont {Lim}, \citenamefont {Moskalets}, \citenamefont
  {Arrachea},\ and\ \citenamefont {S\'anchez}}]{Ludovico2014}%
  \BibitemOpen
  \bibfield  {author} {\bibinfo {author} {\bibfnamefont {M.~F.}\ \bibnamefont
  {Ludovico}}, \bibinfo {author} {\bibfnamefont {J.~S.}\ \bibnamefont {Lim}},
  \bibinfo {author} {\bibfnamefont {M.}~\bibnamefont {Moskalets}}, \bibinfo
  {author} {\bibfnamefont {L.}~\bibnamefont {Arrachea}}, \ and\ \bibinfo
  {author} {\bibfnamefont {D.}~\bibnamefont {S\'anchez}},\ }\href {\doibase
  10.1103/PhysRevB.89.161306} {\bibfield  {journal} {\bibinfo  {journal} {Phys.
  Rev. B}\ }\textbf {\bibinfo {volume} {89}},\ \bibinfo {pages} {161306}
  (\bibinfo {year} {2014})}\BibitemShut {NoStop}%
\bibitem [{\citenamefont {Ludovico}\ \emph {et~al.}(2018)\citenamefont
  {Ludovico}, \citenamefont {Arrachea}, \citenamefont {Moskalets},\ and\
  \citenamefont {S\'anchez}}]{Ludovico2018}%
  \BibitemOpen
  \bibfield  {author} {\bibinfo {author} {\bibfnamefont {M.~F.}\ \bibnamefont
  {Ludovico}}, \bibinfo {author} {\bibfnamefont {L.}~\bibnamefont {Arrachea}},
  \bibinfo {author} {\bibfnamefont {M.}~\bibnamefont {Moskalets}}, \ and\
  \bibinfo {author} {\bibfnamefont {D.}~\bibnamefont {S\'anchez}},\ }\href
  {\doibase 10.1103/PhysRevB.97.041416} {\bibfield  {journal} {\bibinfo
  {journal} {Phys. Rev. B}\ }\textbf {\bibinfo {volume} {97}},\ \bibinfo
  {pages} {041416} (\bibinfo {year} {2018})}\BibitemShut {NoStop}%
\bibitem [{\citenamefont {Bruch}\ \emph {et~al.}(2016)\citenamefont {Bruch},
  \citenamefont {Thomas}, \citenamefont {Viola~Kusminskiy}, \citenamefont {von
  Oppen},\ and\ \citenamefont {Nitzan}}]{Bruch2016}%
  \BibitemOpen
  \bibfield  {author} {\bibinfo {author} {\bibfnamefont {A.}~\bibnamefont
  {Bruch}}, \bibinfo {author} {\bibfnamefont {M.}~\bibnamefont {Thomas}},
  \bibinfo {author} {\bibfnamefont {S.}~\bibnamefont {Viola~Kusminskiy}},
  \bibinfo {author} {\bibfnamefont {F.}~\bibnamefont {von Oppen}}, \ and\
  \bibinfo {author} {\bibfnamefont {A.}~\bibnamefont {Nitzan}},\ }\href
  {\doibase 10.1103/PhysRevB.93.115318} {\bibfield  {journal} {\bibinfo
  {journal} {Phys. Rev. B}\ }\textbf {\bibinfo {volume} {93}},\ \bibinfo
  {pages} {115318} (\bibinfo {year} {2016})}\BibitemShut {NoStop}%
\bibitem [{\citenamefont {Friedel}(1958)}]{Friedel1958}%
  \BibitemOpen
  \bibfield  {author} {\bibinfo {author} {\bibfnamefont {J.}~\bibnamefont
  {Friedel}},\ }\href {\doibase 10.1007/BF02751483} {\bibfield  {journal}
  {\bibinfo  {journal} {Il Nuovo Cimento (1955-1965)}\ }\textbf {\bibinfo
  {volume} {7}},\ \bibinfo {pages} {287} (\bibinfo {year} {1958})}\BibitemShut
  {NoStop}%
\bibitem [{\citenamefont {Bruch}, \citenamefont {Lewenkopf},\ and\
  \citenamefont {von Oppen}(2018)}]{Bruch2018}%
  \BibitemOpen
  \bibfield  {author} {\bibinfo {author} {\bibfnamefont {A.}~\bibnamefont
  {Bruch}}, \bibinfo {author} {\bibfnamefont {C.}~\bibnamefont {Lewenkopf}}, \
  and\ \bibinfo {author} {\bibfnamefont {F.}~\bibnamefont {von Oppen}},\ }\href
  {\doibase 10.1103/PhysRevLett.120.107701} {\bibfield  {journal} {\bibinfo
  {journal} {Phys. Rev. Lett.}\ }\textbf {\bibinfo {volume} {120}},\ \bibinfo
  {pages} {107701} (\bibinfo {year} {2018})}\BibitemShut {NoStop}%
\bibitem [{\citenamefont {Stefanucci}\ and\ \citenamefont {van
  Leeuwen}(2013)}]{Stefanucci13}%
  \BibitemOpen
  \bibfield  {author} {\bibinfo {author} {\bibfnamefont {G.}~\bibnamefont
  {Stefanucci}}\ and\ \bibinfo {author} {\bibfnamefont {R.}~\bibnamefont {van
  Leeuwen}},\ }\href@noop {} {\emph {\bibinfo {title} {Nonequilibrium Many-Body
  Theory Of Quantum Systems: A Modern Introduction}}}\ (\bibinfo  {publisher}
  {Cambridge University Press},\ \bibinfo {year} {2013})\BibitemShut {NoStop}%
\bibitem [{\citenamefont {Bergfield}\ and\ \citenamefont
  {Stafford}(2009)}]{Bergfield09}%
  \BibitemOpen
  \bibfield  {author} {\bibinfo {author} {\bibfnamefont {J.~P.}\ \bibnamefont
  {Bergfield}}\ and\ \bibinfo {author} {\bibfnamefont {C.~A.}\ \bibnamefont
  {Stafford}},\ }\href@noop {} {\bibfield  {journal} {\bibinfo  {journal}
  {Phys. Rev. B}\ }\textbf {\bibinfo {volume} {79}},\ \bibinfo {pages} {245125}
  (\bibinfo {year} {2009})}\BibitemShut {NoStop}%
\bibitem [{\citenamefont {Reuter}\ and\ \citenamefont
  {Hansen}(2014)}]{Reuter14}%
  \BibitemOpen
  \bibfield  {author} {\bibinfo {author} {\bibfnamefont {M.~G.}\ \bibnamefont
  {Reuter}}\ and\ \bibinfo {author} {\bibfnamefont {T.}~\bibnamefont
  {Hansen}},\ }\href {\doibase 10.1063/1.4901722} {\bibfield  {journal}
  {\bibinfo  {journal} {The Journal of Chemical Physics}\ }\textbf {\bibinfo
  {volume} {141}},\ \bibinfo {pages} {181103} (\bibinfo {year}
  {2014})}\BibitemShut {NoStop}%
\bibitem [{\citenamefont {Sam-ang}\ and\ \citenamefont
  {Reuter}(2017)}]{Sam_ang_2017}%
  \BibitemOpen
  \bibfield  {author} {\bibinfo {author} {\bibfnamefont {P.}~\bibnamefont
  {Sam-ang}}\ and\ \bibinfo {author} {\bibfnamefont {M.~G.}\ \bibnamefont
  {Reuter}},\ }\href {\doibase 10.1088/1367-2630/aa6c23} {\bibfield  {journal}
  {\bibinfo  {journal} {New Journal of Physics}\ }\textbf {\bibinfo {volume}
  {19}},\ \bibinfo {pages} {053002} (\bibinfo {year} {2017})}\BibitemShut
  {NoStop}%
\bibitem [{\citenamefont {Gasparian}, \citenamefont {Christen},\ and\
  \citenamefont {B\"uttiker}(1996)}]{Gasparian96}%
  \BibitemOpen
  \bibfield  {author} {\bibinfo {author} {\bibfnamefont {V.}~\bibnamefont
  {Gasparian}}, \bibinfo {author} {\bibfnamefont {T.}~\bibnamefont {Christen}},
  \ and\ \bibinfo {author} {\bibfnamefont {M.}~\bibnamefont {B\"uttiker}},\
  }\href@noop {} {\bibfield  {journal} {\bibinfo  {journal} {Phys. Rev. A}\
  }\textbf {\bibinfo {volume} {54}},\ \bibinfo {pages} {4022} (\bibinfo {year}
  {1996})}\BibitemShut {NoStop}%
\bibitem [{\citenamefont {Gramespacher}\ and\ \citenamefont
  {B\"uttiker}(1997)}]{Gramespacher97}%
  \BibitemOpen
  \bibfield  {author} {\bibinfo {author} {\bibfnamefont {T.}~\bibnamefont
  {Gramespacher}}\ and\ \bibinfo {author} {\bibfnamefont {M.}~\bibnamefont
  {B\"uttiker}},\ }\href {\doibase 10.1103/PhysRevB.56.13026} {\bibfield
  {journal} {\bibinfo  {journal} {Phys. Rev. B}\ }\textbf {\bibinfo {volume}
  {56}},\ \bibinfo {pages} {13026} (\bibinfo {year} {1997})}\BibitemShut
  {NoStop}%
\bibitem [{\citenamefont {Stafford}(2016)}]{Stafford16}%
  \BibitemOpen
  \bibfield  {author} {\bibinfo {author} {\bibfnamefont {C.~A.}\ \bibnamefont
  {Stafford}},\ }\href {\doibase 10.1103/PhysRevB.93.245403} {\bibfield
  {journal} {\bibinfo  {journal} {Phys. Rev. B}\ }\textbf {\bibinfo {volume}
  {93}},\ \bibinfo {pages} {245403} (\bibinfo {year} {2016})}\BibitemShut
  {NoStop}%
\bibitem [{\citenamefont {Stafford}, \citenamefont {Baeriswyl},\ and\
  \citenamefont {B{\"u}rki}(1997)}]{Stafford97a}%
  \BibitemOpen
  \bibfield  {author} {\bibinfo {author} {\bibfnamefont {C.~A.}\ \bibnamefont
  {Stafford}}, \bibinfo {author} {\bibfnamefont {D.}~\bibnamefont {Baeriswyl}},
  \ and\ \bibinfo {author} {\bibfnamefont {J.}~\bibnamefont {B{\"u}rki}},\
  }\href@noop {} {\bibfield  {journal} {\bibinfo  {journal} {Phys. Rev. Lett.}\
  }\textbf {\bibinfo {volume} {79}},\ \bibinfo {pages} {2863} (\bibinfo {year}
  {1997})}\BibitemShut {NoStop}%
\bibitem [{\citenamefont {Stafford}\ and\ \citenamefont
  {Shastry}(2017)}]{Stafford17}%
  \BibitemOpen
  \bibfield  {author} {\bibinfo {author} {\bibfnamefont {C.~A.}\ \bibnamefont
  {Stafford}}\ and\ \bibinfo {author} {\bibfnamefont {A.}~\bibnamefont
  {Shastry}},\ }\href {\doibase 10.1063/1.4975810} {\bibfield  {journal}
  {\bibinfo  {journal} {The Journal of Chemical Physics}\ }\textbf {\bibinfo
  {volume} {146}},\ \bibinfo {pages} {092324} (\bibinfo {year}
  {2017})}\BibitemShut {NoStop}%
\bibitem [{\citenamefont {Shastry}(2018)}]{Shastry18}%
  \BibitemOpen
  \bibfield  {author} {\bibinfo {author} {\bibfnamefont {A.~S.~C.}\
  \bibnamefont {Shastry}},\ }\emph {\bibinfo {title} {Theory of Thermodynamic
  Measurements of Quantum Systems Far from Equilibrium}},\ \href
  {http://ezproxy.library.arizona.edu/login?url=https://search.proquest.com/docview/2049676197?accountid=8360}
  {Ph.D. thesis},\ \bibinfo  {school} {University of Arizona} (\bibinfo {year}
  {2018}),\ \bibinfo {note} {copyright - Database copyright ProQuest LLC;
  ProQuest does not claim copyright in the individual underlying works; Last
  updated - 2018-07-04}\BibitemShut {NoStop}%
\bibitem [{\citenamefont {Shastry}\ and\ \citenamefont
  {Stafford}(2019)}]{Shastry19}%
  \BibitemOpen
  \bibfield  {author} {\bibinfo {author} {\bibfnamefont {A.}~\bibnamefont
  {Shastry}}\ and\ \bibinfo {author} {\bibfnamefont {C.~A.}\ \bibnamefont
  {Stafford}},\ }\href@noop {} {}\bibinfo {howpublished} {unpublished}
  (\bibinfo {year} {2019})\BibitemShut {NoStop}%
\bibitem [{\citenamefont {Shastry}\ and\ \citenamefont
  {Stafford}(2015)}]{Shastry15}%
  \BibitemOpen
  \bibfield  {author} {\bibinfo {author} {\bibfnamefont {A.}~\bibnamefont
  {Shastry}}\ and\ \bibinfo {author} {\bibfnamefont {C.~A.}\ \bibnamefont
  {Stafford}},\ }\href {\doibase 10.1103/PhysRevB.92.245417} {\bibfield
  {journal} {\bibinfo  {journal} {Phys. Rev. B}\ }\textbf {\bibinfo {volume}
  {92}},\ \bibinfo {pages} {245417} (\bibinfo {year} {2015})}\BibitemShut
  {NoStop}%
\bibitem [{\citenamefont {Shastry}\ and\ \citenamefont
  {Stafford}(2016)}]{Shastry16}%
  \BibitemOpen
  \bibfield  {author} {\bibinfo {author} {\bibfnamefont {A.}~\bibnamefont
  {Shastry}}\ and\ \bibinfo {author} {\bibfnamefont {C.~A.}\ \bibnamefont
  {Stafford}},\ }\href {\doibase 10.1103/PhysRevB.94.155433} {\bibfield
  {journal} {\bibinfo  {journal} {Phys. Rev. B}\ }\textbf {\bibinfo {volume}
  {94}},\ \bibinfo {pages} {155433} (\bibinfo {year} {2016})}\BibitemShut
  {NoStop}%
\bibitem [{\citenamefont {Ness}(2014)}]{Ness14}%
  \BibitemOpen
  \bibfield  {author} {\bibinfo {author} {\bibfnamefont {H.}~\bibnamefont
  {Ness}},\ }\href {\doibase 10.1103/PhysRevB.89.045409} {\bibfield  {journal}
  {\bibinfo  {journal} {Phys. Rev. B}\ }\textbf {\bibinfo {volume} {89}},\
  \bibinfo {pages} {045409} (\bibinfo {year} {2014})}\BibitemShut {NoStop}%
\bibitem [{\citenamefont {Dubi}\ and\ \citenamefont
  {Di~Ventra}(2009)}]{DiVentra09}%
  \BibitemOpen
  \bibfield  {author} {\bibinfo {author} {\bibfnamefont {Y.}~\bibnamefont
  {Dubi}}\ and\ \bibinfo {author} {\bibfnamefont {M.}~\bibnamefont
  {Di~Ventra}},\ }\href@noop {} {\bibfield  {journal} {\bibinfo  {journal}
  {Nano Letters}\ }\textbf {\bibinfo {volume} {9}},\ \bibinfo {pages} {97}
  (\bibinfo {year} {2009})}\BibitemShut {NoStop}%
\bibitem [{\citenamefont {Jacquet}\ and\ \citenamefont
  {Pillet}(2012)}]{Jacquet2012}%
  \BibitemOpen
  \bibfield  {author} {\bibinfo {author} {\bibfnamefont {P.~A.}\ \bibnamefont
  {Jacquet}}\ and\ \bibinfo {author} {\bibfnamefont {C.-A.}\ \bibnamefont
  {Pillet}},\ }\href {\doibase 10.1103/PhysRevB.85.125120} {\bibfield
  {journal} {\bibinfo  {journal} {Phys. Rev. B}\ }\textbf {\bibinfo {volume}
  {85}},\ \bibinfo {pages} {125120} (\bibinfo {year} {2012})}\BibitemShut
  {NoStop}%
\bibitem [{\citenamefont {Bergfield}\ \emph {et~al.}(2013)\citenamefont
  {Bergfield}, \citenamefont {Story}, \citenamefont {Stafford},\ and\
  \citenamefont {Stafford}}]{Bergfield2013demon}%
  \BibitemOpen
  \bibfield  {author} {\bibinfo {author} {\bibfnamefont {J.~P.}\ \bibnamefont
  {Bergfield}}, \bibinfo {author} {\bibfnamefont {S.~M.}\ \bibnamefont
  {Story}}, \bibinfo {author} {\bibfnamefont {R.~C.}\ \bibnamefont {Stafford}},
  \ and\ \bibinfo {author} {\bibfnamefont {C.~A.}\ \bibnamefont {Stafford}},\
  }\href {\doibase 10.1021/nn401027u} {\bibfield  {journal} {\bibinfo
  {journal} {ACS Nano}\ }\textbf {\bibinfo {volume} {7}},\ \bibinfo {pages}
  {4429} (\bibinfo {year} {2013})}\BibitemShut {NoStop}%
\bibitem [{\citenamefont {Meair}\ \emph {et~al.}(2014)\citenamefont {Meair},
  \citenamefont {Bergfield}, \citenamefont {Stafford},\ and\ \citenamefont
  {Jacquod}}]{Meair14}%
  \BibitemOpen
  \bibfield  {author} {\bibinfo {author} {\bibfnamefont {J.}~\bibnamefont
  {Meair}}, \bibinfo {author} {\bibfnamefont {J.~P.}\ \bibnamefont
  {Bergfield}}, \bibinfo {author} {\bibfnamefont {C.~A.}\ \bibnamefont
  {Stafford}}, \ and\ \bibinfo {author} {\bibfnamefont {P.}~\bibnamefont
  {Jacquod}},\ }\href {\doibase 10.1103/PhysRevB.90.035407} {\bibfield
  {journal} {\bibinfo  {journal} {Phys. Rev. B}\ }\textbf {\bibinfo {volume}
  {90}},\ \bibinfo {pages} {035407} (\bibinfo {year} {2014})}\BibitemShut
  {NoStop}%
\bibitem [{\citenamefont {Bergfield}\ \emph {et~al.}(2015)\citenamefont
  {Bergfield}, \citenamefont {Ratner}, \citenamefont {Stafford},\ and\
  \citenamefont {Di~Ventra}}]{Bergfield15}%
  \BibitemOpen
  \bibfield  {author} {\bibinfo {author} {\bibfnamefont {J.~P.}\ \bibnamefont
  {Bergfield}}, \bibinfo {author} {\bibfnamefont {M.~A.}\ \bibnamefont
  {Ratner}}, \bibinfo {author} {\bibfnamefont {C.~A.}\ \bibnamefont
  {Stafford}}, \ and\ \bibinfo {author} {\bibfnamefont {M.}~\bibnamefont
  {Di~Ventra}},\ }\href {\doibase 10.1103/PhysRevB.91.125407} {\bibfield
  {journal} {\bibinfo  {journal} {Phys. Rev. B}\ }\textbf {\bibinfo {volume}
  {91}},\ \bibinfo {pages} {125407} (\bibinfo {year} {2015})}\BibitemShut
  {NoStop}%
\bibitem [{\citenamefont {B\"uttiker}(1989)}]{Buttiker89}%
  \BibitemOpen
  \bibfield  {author} {\bibinfo {author} {\bibfnamefont {M.}~\bibnamefont
  {B\"uttiker}},\ }\href {\doibase 10.1103/PhysRevB.40.3409} {\bibfield
  {journal} {\bibinfo  {journal} {Phys. Rev. B}\ }\textbf {\bibinfo {volume}
  {40}},\ \bibinfo {pages} {3409} (\bibinfo {year} {1989})}\BibitemShut
  {NoStop}%
\bibitem [{\citenamefont {Bergfield}\ and\ \citenamefont
  {Stafford}(2014)}]{Bergfield14}%
  \BibitemOpen
  \bibfield  {author} {\bibinfo {author} {\bibfnamefont {J.~P.}\ \bibnamefont
  {Bergfield}}\ and\ \bibinfo {author} {\bibfnamefont {C.~A.}\ \bibnamefont
  {Stafford}},\ }\href {\doibase 10.1103/PhysRevB.90.235438} {\bibfield
  {journal} {\bibinfo  {journal} {Phys. Rev. B}\ }\textbf {\bibinfo {volume}
  {90}},\ \bibinfo {pages} {235438} (\bibinfo {year} {2014})}\BibitemShut
  {NoStop}%
\bibitem [{\citenamefont {Barr}, \citenamefont {Stafford},\ and\ \citenamefont
  {Bergfield}(2012)}]{Barr12}%
  \BibitemOpen
  \bibfield  {author} {\bibinfo {author} {\bibfnamefont {J.~D.}\ \bibnamefont
  {Barr}}, \bibinfo {author} {\bibfnamefont {C.~A.}\ \bibnamefont {Stafford}},
  \ and\ \bibinfo {author} {\bibfnamefont {J.~P.}\ \bibnamefont {Bergfield}},\
  }\href {\doibase 10.1103/PhysRevB.86.115403} {\bibfield  {journal} {\bibinfo
  {journal} {Phys. Rev. B}\ }\textbf {\bibinfo {volume} {86}},\ \bibinfo
  {pages} {115403} (\bibinfo {year} {2012})}\BibitemShut {NoStop}%
\end{thebibliography}%

\end{document}